\documentclass[a4paper,12pt]{article}
\usepackage{amsmath}
\usepackage{amssymb}
\usepackage{amsfonts,times}
\usepackage{theorem}
\usepackage{hyperref}
\AtBeginDocument{}
\newtheorem{theorem}{Theorem}

\newtheorem{definition}[theorem]{Definition}

\newtheorem{lemma}[theorem]{Lemma}

\newtheorem{proposition}[theorem]{Proposition}
{\theorembodyfont{\upshape}\newtheorem{remark}[theorem]{Remark}}
{\theorembodyfont{\upshape}}

\newenvironment{proof}[1][Proof]{\noindent\textbf{#1.} }{\ \hfill \rule{0.5em}{0.5em}}

\newcommand{\IN}{\mathbb{N}}
\newcommand{\IE}{\mathbb{E}}
\newcommand{\IP}{\mathbb{P}}
\newcommand{\IR}{\mathbb{R}}

\begin{document}

\title{Semicircle law for generalized Curie-Weiss matrix ensembles at subcritical temperature}
\author{Werner Kirsch \\
Fakult\"{a}t f\"{u}r Mathematik und Informatik\\
FernUniversit\"{a}t in Hagen, Germany \and Thomas Kriecherbauer \\
Mathematisches Institut\\
Universit\"{a}t Bayreuth, Germany}
\date{}

\maketitle

\abstract{In \cite{HKW} Hochst\"attler, Kirsch, and Warzel showed that the semicircle law holds for generalized Curie-Weiss matrix ensembles
at or above the critical temperature. We extend their result to the case of subcritical temperatures for which the
correlations between the matrix entries are stronger. Nevertheless, one may use the concept of {\em approximately
uncorrelated} ensembles that was first introduced in \cite{HKW}. In order to do so one needs to remove the
average magnetization of the entries by an appropriate modification of the ensemble that turns out to be of rank $1$
thus not changing the limiting spectral measure.
}

\section{Introduction}

Hochst\"attler, Kirsch, and Warzel proved in \cite{HKW} the semicircle law for ensembles of real symmetric matrices
where the upper triangular part is filled by what they called
{\em approximately uncorrelated} random variables \cite[Def. 4]{HKW}  (see also Definition \ref{def:approx} below).
An important motivation for introducing this notion is that collections of random variables with values in $\{-1,1\}$
that are distributed according to the Curie-Weiss law at or above the critical temperature are approximately uncorrelated and thus
the semicircle law holds for the corresponding matrix ensembles. It is the main goal of the present paper to show how one
may use the concept of approximately uncorrelated random variables to prove a semicircle law also for subcritical temperatures.

In order to state our result precisely we need a few definitions.
Curie-Weiss random variables $\xi_1, \ldots , \xi_M$, also called spins, take values in $\{-1,1\}$ with probability
\begin{equation*}
\IP_{\beta}^M (\xi_1 = x_1, \ldots, \xi_M = x_M) = Z_{\beta,M}^{-1} e^{\frac{\beta}{2M} (\sum_{i=1}^M x_i)^2}
\end{equation*}
where $Z_{\beta,M}$ denotes the normalization constant. By $\IE_{\beta}^{M}$ we denote the expectation with respect to $\IP_{\beta}^{M}$.

The parameter $\beta \geq 0$ is interpreted in physics as
{\em inverse temperature}, $\beta = \frac{1}{T}$.
Information on the physical meaning of the Curie-Weiss model can be found
in \cite{Ellis, Thompson}. If $\beta = 0$ the random variables $\xi_i$ are independent while
for $\beta > 0$ there is a positive correlation between the $\xi_i$ that grows with $\beta$.
At the {\em critical inverse temperature} $\beta=1$ a phase transition occurs.
While at and above the critical temperature $(\beta \leq 1)$ the average spin $\frac{1}{M} \sum_{i=1}^M \xi_i$
converges in distribution to the Dirac measure $\delta_0$, the average spin converges to
$\frac{1}{2} (\delta_{-m (\beta)} + \delta_{m (\beta)})$ below the critical temperature
$(\beta >1)$, where $m = m (\beta) \in (0,1)$ is called the {\em average magnetization}.
It can be defined as the (unique) strictly positive solution of
\begin{equation}\label{eq:defm}
\tanh (\beta m) = m\,.
\end{equation}
For a proof of this fact see e.\,g. \cite{Ellis} or \cite{Kirsch}. The following observation is fundamental for the analysis of
\cite{HKW} and also for the present paper: Curie-Weiss distributed random variables are of {\em de Finetti type}, i.e. they can be
represented as an average of independently distributed random variables. More precisely, for $t \in [-1,1]$ we denote
the probability measures
\begin{eqnarray}\nonumber
P_t^{(M)} &=& \textstyle \bigotimes_{j=1}^M P_t^{(1)}  \text{~  as the $M$-fold product of $P_t^{(1)}$ with}\\
P_t^{(1)} (\{1\}) &=& \textstyle \frac{1}{2} (1+t) \text{~  and ~} P_t^{(1)} (\{-1\}) \, = \, \frac{1}{2} (1-t)\,.\label{def:Pt}
\end{eqnarray}
If $M$ is clear from the context we write $P_t$ instead of $P_t^{(M)}$. By $E_t$ resp. $E_t^{(M)}$
we denote the corresponding expectation. For any function $\phi$ on $\{-1,1\}^M$ we have
\begin{eqnarray}\label{eq:dFCW}\nonumber
&&\IE_{\beta}^M \Big(\phi (X_1, \ldots, X_M) \Big) = \frac{1}{Z}\int_{-1}^1 E_t \Big( \phi (X_1, \ldots , X_M) \Big)
\frac{e^{-M F_{\beta} (t)/2}}{1-t^2} dt \,, \\ \label{def:F}
&&\text{with ~~~~} F_{\beta} (t) := \frac{1}{\beta} \Big(\frac{1}{2} \ln \frac{1+t}{1-t}\Big)^2 + \ln (1-t^2)\,, \quad t \in (-1,1)\,,
\vspace{-0.3cm}
\end{eqnarray}
and normalization constant $Z:=\int_{-1}^1e^{-M F_{\beta} (t)/2}\frac{dt}{1-t^2}$.

In this sense $\IP_{\beta}^M$ is
a weighted $t$ average over all $P_t^{(M)}$.
This fact can be proved using the so called Hubbard-Stratonovich transformation (see \cite{HKW} or \cite{Kirsch}
and references therein). The above considerations motivate the following definition of
{\em generalized Curie-Weiss ensembles} introduced in \cite[Def. 29]{HKW}.
\begin{definition}\label{def:GCW}
Suppose $\alpha>0$, $\beta \geq 0$ and let $F_{\beta}$, $ P_t^{(M)}$ be defined as in \eqref{def:F}, \eqref{def:Pt} above.
We define a probability measure on $\mathbb{R}^{N^2}$ by
\begin{equation*}
\IP_{N^2}^{N^{\alpha} F_{\beta}} = \int_{-1}^1 P_t^{(N^2)} ~  d \nu_N (t) ~  , ~~  d \nu_N (t) =
\frac{1}{Z_N} \cdot \frac{e^{-N^{\alpha} F_{\beta} (t)/2}}{1 - t^2} dt\,,
\end{equation*}
where $Z_N$ is chosen such that the {\em de Finetti measure} $\nu_N$ becomes a probability measure
(see Proposition \ref{prop:F} c) for normalizability of $\nu_N$).
The corresponding matrix ensemble $(X_N)_N$ is then defined as follows. Pick $\frac{N(N+1)}{2}$ different components of the
$\IP_{N^2}^{N^{\alpha} F_{\beta}}$- distributed random vector to fill the upper triangular part of $X_N \in \mathbb{R}^{N \times N}$.
The remaining entries $X_N(i,j)$, $1 \leq j < i \leq N$ are then determined by the symmetry $X_N(i,j)=X_N(j,i)$.
Observe that $\IP_{N^2}^{N^{\alpha} F_{\beta}}$ is invariant under permutations of components.
Thus the resulting matrix ensemble does neither depend on the choice of the $\frac{N(N+1)}{2}$
components (out of $N^2$) nor on the
order in which the upper triangular part of $X_N$ is filled.
The so generated random matrix ensemble $(X_{N})_N$ is called a generalized
 $\mathbb{P}_{N^{2}}^{N^{\alpha } F_{\beta}}$- Curie-Weiss ensemble.
\end{definition}

\begin{remark}
Note that it is only in the case $\alpha = 2$ that the $N^2$ real random variables from which the entries of the matrix $X_N$
are chosen are indeed Curie-Weiss distributed. The generalization
thus consists of allowing for $\alpha \neq 2$ which still remains in the framework of exchangeable random variables.
However, for $\alpha > 2$ one may view the ensemble to be generated by selecting the matrix entries from a
collection of $N^{\alpha}$ Curie-Weiss distributed random spins $\pm 1$.

Observe that the definition in \cite{HKW} is even more general than
Definition \ref{def:GCW}. It also allows
to replace $F_{\beta}$ by more general functions $F$ (cf. Remark \ref{remark:generalF}).
\end{remark}

Let $\IP$ be the probability measure underlying some matrix ensemble $(A_N)_N$ of real symmetric matrices
$A_N \in \mathbb{R}^{N\times N}$.
Denote by $\lambda_1 \leq \ldots \leq \lambda_N$ the eigenvalues of $A_N$.
Then we call $\sigma_N := \frac{1}{N} \sum_{i=1}^N \delta_{\lambda_i}$ the eigenvalue distribution measure of $A_N$.
We say that the {\em semicircle law holds} for the ensemble $(A_N)_N$, if $\sigma_N$ converges weakly in probability
to $\sigma_{sc}$ as $N \to \infty$, i.e. for all bounded continuous functions $f: \IR \rightarrow \IR$
and for all $\epsilon > 0$ we have
\begin{equation*}
\underset{N \rightarrow \infty}{\lim} \IP \bigg{(} \Big{|}
\int_{-\infty}^{\infty} f(x) ~  d \sigma_N (x) - \int_{-\infty}^{\infty} f(x) ~  d \sigma_{sc} (x)
\Big{|} > \epsilon \bigg{)} = 0 \,,
\end{equation*}
where the semicircle
$\sigma_{sc}$ is the Borel measure on $\mathbb{R}$ with support $[-2,2]$ and (Lebesgue-) density  $\frac{1}{2\pi}
\sqrt{4-x^2}$ at $x \in [-2,2]$.

The following result is proved in \cite{HKW}.
\begin{theorem}[Theorem 31 and Corollary 28 in \protect\cite{HKW}]\label{thm:HKW}
Let $(X_{N})_N$ be a generalized $\mathbb{P}_{N^{2}}^{N^{\alpha } F_{\beta}}$- Curie-Weiss ensemble. Then the semicircle law holds
for $X_N/\sqrt{N}$ if \vspace{0.1cm} either

(i) \quad $\beta \in [0,1)\,$ and $\,\alpha\geq 1$ \qquad or \qquad  (ii) \quad $\beta =1\,$ and $\,\alpha\geq 2$ \\
is satisfied.
\end{theorem}

The purpose of the present paper is to extend this result to subcritical temperatures $\beta > 1$. Our main result reads:
\begin{theorem}
\label{theo:A1}

Let $(X_N)_N$ be a generalized $\IP_{N^2}^{N^{\alpha} F_{\beta}}$- Curie-Weiss ensemble with $\alpha \geq 1$ and $\beta > 1$.
Denote by $m = m (\beta)$ the unique strictly positive solution of the equation $\tanh (\beta m) = m$ $($cf. \eqref{eq:defm}$)$.
Then the semicircle law holds for $A_N := X_N/\sqrt{N (1 - m^2)}$, i.e.
for all bounded continuous functions $f: \IR \rightarrow \IR$ and for all $\epsilon > 0$ we have
\begin{equation*}
\underset{N \rightarrow \infty}{\lim} \IP_{N^2}^{N^{\alpha} F_{\beta}} \bigg{(}
\Big{|} \int_{-\infty}^{\infty} f(x) ~  d \sigma_N (x) - \int_{-\infty}^{\infty} f(x) ~  d \sigma_{sc} (x) \Big{|} > \epsilon \bigg{)} = 0 \,,
\end{equation*}
where $\sigma_N$ denotes the eigenvalue distribution measure of $A_N$ .
\end{theorem}

As mentioned above, Theorem \ref{thm:HKW} was proved in \cite{HKW} by introducing the concept of
{\em approximately uncorrelated random variables}. For ensembles of real symmetric matrices this
may be formulated as follows (cf. \cite[Def. 4]{HKW}).

\begin{definition}\label{def:approx}
Let $(\Omega, \mathcal{F}, \IP)$ denote the probability space for an ensemble $(X_N)_N$ of real symmetric
$N \times N$ matrices. We say that the entries are approximately uncorrelated if
\begin{equation}
\label{eq:A7.1}
\bigg{|}\IE \ \Big{(}\prod_{\nu = 1}^{\ell} X_N (i_{\nu}, j_{\nu})
\prod^p_{\rho = 1} X_N (u_{\rho}, v_{\rho}) \Big{)} \bigg{|} \leq \frac{C_{\ell,p}}{N^{\frac{\ell}{2}}}
\end{equation}
and for every $\ell\in\IN$ there is a sequence $a_{\ell, N}$ converging to $0$ as $N\to\infty$ with 
\begin{equation}
\label{eq:A7.2}
\Big|\IE \ \Big{(}\prod^{\ell}_{\nu = 1} X_N (i_{\nu}, j_{\nu})^2 \Big{)} - 1\Big|~\leq~a_{\ell, N}
\end{equation}
for all sequences $(i_1, j_1), \ldots , (i_{\ell}, j_{\ell})$ which are pairwise disjoint and disjoint to the sequences
$(u_1, v_1), \ldots , (u_p, v_p)$ with $N$-independent constants $C_{\ell,p}$. By sequence we mean that all indices
$(i_s, j_s)$, $(u_s,v_s)$ may depend on $N$ and that for each value of $N$ they belong to the set $\{ (i,j) : 1\leq i \leq j \leq N \}$.
\end{definition}

In \cite[Theorem 5]{HKW} it is shown that conditions \eqref{eq:A7.1} and \eqref{eq:A7.2} are well suited to apply the original ideas
of Wigner \cite{Wigner1, Wigner2} and Grenander \cite{Grenander} to prove the semicircle law via the method of moments
(see also \cite{Arnold, Arnold71, MarchenkoP, Pastur72, Pastur} and the
monographs \cite{AGZ, Bai etal, Pastur Sherbina, Tao}). Matrix ensembles with correlated entries have already been considered in
\cite{Bryc et al, FriesenLoewe1, FriesenLoewe2, Goetze Naumov Tikhomirov, Goetze Tikhomirov 1, HofmannS, LoeweS,
Schenker Schulz-Baldes}. See \cite{HKW} and the recent survey \cite{KK1} for more information on these results.

Theorem \ref{thm:HKW} follows from \cite[Theorem 5]{HKW} by verifying that under conditions (i) and (ii) the
ensemble $\IP_{N^2}^{N^{\alpha} F_{\beta}}$ has approximately
uncorrelated entries. For $\beta > 1$ the situation is different.
For example, condition \eqref{eq:A7.1} is violated for any $\alpha > 0$, because (cf. \cite[proof of Proposition 35]{HKW})
\begin{equation}\label{eq:counter}
\lim_{N \to \infty} \IE_{N^2}^{N^{\alpha} F_{\beta}} \big{(} X_N (1, 1) X_N (1, 2) \big{)} = m^2 > 0\,,
\end{equation}
where, again, $m$ denotes the average magnetization \eqref{eq:defm}.
Indeed, by the independence of $X_N (1, 1)$ and $X_N (1, 2)$ with
respect to $P_t$ we conclude $E_t (X_N (1, 1) X_N (1, 2)) =E_t (X_N (1, 1))E_t (X_N (1, 2))=t^2$.
Moreover, and this is the crucial difference from the
case $\beta \leq 1$, the function $F_{\beta}$ has two minimizers $\pm m$ (cf. Proposition \ref{prop:F})
so that for large values of $N$ we have
$\nu_N \sim \frac{1}{2} (\delta_{-m} + \delta_{m})$ from which \eqref{eq:counter} follows.
Proposition \ref{prop:Laplace} stated below can be used to make this argument rigorous.

In order to prove Theorem \ref{theo:A1} we use an idea that has already been introduced in the proof of Proposition 32 in
\cite{HKW}. For the moment let us proceed heuristically and assume $\nu_N = \frac{1}{2} (\delta_{-m} + \delta_m)$. Therefore
we only need to consider the matrix ensembles related to $P_{\pm m}$.
These fall into the class of real symmetric matrices with i.i.d. entries $X_N (i,j), i \leq j$
that are distributed according to $\frac{1}{2} [(1 \pm m) \delta_1 + (1 \mp m) \delta_{-1}]$.
Consequently $E_{\pm m} (X_N (i,j)) = \pm m$ and $E_{\pm m} (X_N^2 (i,j)) = 1$.
Subtracting the mean and dividing by the standard deviation leads to standard Wigner ensembles
\begin{equation}
\label{eq:A2.2}
\tilde{X}_{N,\pm} := \frac{1}{\sqrt{1 - m^2}} (X_N \mp m \mathcal{E}_N), ~~  \mathcal{E}_N (i,j) := 1
\text{ ~ for all ~} 1\leq i,j \leq N \,.
\end{equation}
By the classical results of Wigner \cite{Wigner1, Wigner2} the eigenvalue distribution measures of
$\tilde{X}_{N\pm} / \sqrt{N}$ converge to the semicircle law $\sigma_{sc}$ when considered
with respect to the probability measures $P_{\pm m}$.

Now we need to relate $\tilde{X}_{N,\pm}$ back to $X_N$. To this end, observe that $X_N/\sqrt{N (1 - m^2)}$
is a rank 1 perturbation of both $\tilde{X}_{N,\pm} / \sqrt{N}$ and we can therefore expect (see Lemma \ref{lemma:perturb})
that the eigenvalue distribution measures of $X_N/\sqrt{N (1 - m^2)}$ converge to the semicircle  law as well.
What is still missing is an indicator when to replace the $\mathbb{P}_{N^{2}}^{N^{\alpha } F_{\beta}}$- distributed $X_N$
by $\tilde{X}_{N,+}$ and when by $\tilde{X}_{N,-}$. As we will see by some basic large deviations argument in Proposition
\ref{prop:A4} the sum of the entries is an efficient choice for that.
Define
\begin{equation}
\label{eq:A3.1}
S_N := \!\!\!\! \sum_{1 \leq i \leq j \leq N} \!\!\!\! X_N (i,j), ~~
\mathcal{X}_{N,+} := 1_{\{S_N > 0 \}}, ~~  \mathcal{X}_{N,-} := 1_{\{S_N \leq 0 \}} = 1 - \mathcal{X}_{N,+} \,,
\end{equation}
\begin{equation}
\label{eq:A4.1}
Y_{N,\pm} := \frac{1}{\sqrt{1 - m^2}} (X_N \mp m \mathcal{E}_N) \mathcal{X}_{N,\pm}\,, ~~~  Y_N := Y_{N,+} + Y_{N,-} \,\,.
\end{equation}
The core of the proof of Theorem \ref{theo:A1} is to show that the entries of $Y_N$ are approximately uncorrelated. Since
$Y_N - X_N/\sqrt{1 - m^2}$ has rank $1$ our main result is then a consequence of
\cite[Theorem 5]{HKW} and Lemma \ref{lemma:perturb}.

We also have results on the largest and second largest
singular value of generalized Curie-Weiss matrices \cite{KK3}, which complement and improve
on \cite{HKW}.

The plan of the paper is as follows.
In Section \ref{sec.2} we use Laplace's method to make the starting point of our heuristic argument,
$\nu_N \sim \frac{1}{2} (\delta_{-m} + \delta_{m})$ for large values of $N$, precise.
The remaining arguments for the proof of our main result Theorem \ref{theo:A1} are gathered in Section \ref{sec.3}.

\medskip \textbf{Acknowledgement}
The second author would like to thank the Lehrgebiet Stochastics at the FernUniversit\"{a}t in Hagen,
where most of the work was accomplished, for support and great hospitality.

The authors are grateful to Michael Fleermann for valuable suggestions.

\section{Analysis of the de Finetti measures}
\label{sec.2}

In this section we state and prove in Lemma \ref{lemma:A5} the precise version of the formula
$\nu_N \sim \frac{1}{2} (\delta_{-m} + \delta_{m})$ for the de Finetti measure
that we used in the heuristic arguments at the end of the Introduction.

\begin{lemma}
\label{lemma:A5}

Assume $\alpha > 0$, $\beta > 1$ and let $\nu_N \equiv \nu_{N, \alpha, \beta}$ be given as in Definition \ref{def:GCW}.
Recall also the definition of the average magnetization $m \equiv m (\beta)$ in \eqref{eq:defm} . Then the following holds:
\renewcommand{\labelenumi}{\alph{enumi})}
\begin{enumerate}
\item There exist numbers $C, \delta > 0$ such that for all $N \in \IN:$
\begin{equation*}
\nu_N ([-\frac{m}{2}, \frac{m}{2}]) \leq C e^{-\delta N^{\alpha}}
\end{equation*}
\item For all $\ell \in \IN$ there exist $C_{\ell} > 0$ such that for all $N \in \IN:$
\begin{equation*}
\int_{\frac{m}{2}}^1  |t-m|^{\ell} ~ d \nu_N (t) \leq C_{\ell} N^{-\frac{\alpha \ell}{2}}
\end{equation*}
\item $\displaystyle \underset{N \rightarrow \infty}{\lim} \int_{\frac{m}{2}}^1 (1 + m^2 - 2mt)^{\ell} ~ d \nu_N (t) =
\frac{1}{2} (1 - m^2)^{\ell}$ ~~  for all $\ell \in \IN$.
\end{enumerate}
\end{lemma}

Lemma \ref{lemma:A5} is proved at the end of this section using Laplace's method. The following version is a special case
of \cite[Theorem 7.1]{Olver} (cf. \cite[Proposition 24]{HKW}).

\begin{proposition}[Laplace method \protect\cite{Olver}]
\label{prop:Laplace}Suppose ${F:}\left( -1,1\right) \rightarrow \mathbb{R}$
is differentiable, $\phi :\left( -1,1\right) \rightarrow \mathbb{R}$ is continuous and for some $-1 <a <b \leq 1$ we have

\begin{enumerate}
\item  $\inf_{t\in \left[
c,b\right) }F(t)>F(a)\,$ for all $c\in \left( a,b\right) $.


\item As $t\searrow a$ we have
\begin{eqnarray}
F(t) &=&F(a)+P\,\left( t-a\right) ^{\kappa }\,+\,\mathcal{O}\big(\left( t-a\right)
^{\kappa +1}\big)  \label{eq:Olver1} \\
\phi (t) &=&Q\,\left( t-a\right) ^{\lambda -1}\,+\,\mathcal{O}\big(\left(
t-a\right) ^{\lambda }\big)  \label{eq:Olver2}
\end{eqnarray}

where $\kappa ,\lambda $ and $P$ are positive constants, $Q$ is a real
constant, and (\ref{eq:Olver1}) is differentiable.

\item For $x \in \mathbb{R}$ sufficiently large the function $t \mapsto e^{-x\,F(t)\,/2}\,\phi(t)\,$ belongs to $L^1(a,b)$.
\end{enumerate}

\noindent
Then as $x\rightarrow \infty\,$ the integral $I(x) := \int_a^b e^{-x\,F(t)\,/2}\,\phi(t)\,dt$ satisfies
\begin{equation*}
I\left( x\right) ~\approx ~\frac{Q}{\kappa }\;\;\Gamma\!\left( \frac{\lambda }{\kappa
}\right) \,\left( \frac{2}{xP}\right) ^{\frac{
\lambda }{\kappa }}\,e^{-x\,F(a)/2}
\end{equation*}
where $A(x)\approx B(x)$ means $\lim_{x\to\infty} \frac{A(x)}{B(x)}=1$ and\, $\Gamma$ denotes the Gamma function.
\end{proposition}

Next we summarize those properties of the function $F_{\beta}$ that are used in the proof of Lemma \ref{lemma:A5}.
In view of Proposition \ref{prop:Laplace} we need to analyze the monotonicity properties of $F_{\beta}$. In addition we provide
an estimate that is useful to establish integrability of $\nu_N$ near the endpoints $t=\pm 1$.

\begin{proposition}
\label{prop:F}
Assume $\beta > 1$ and let $F_{\beta}$, $m \equiv m (\beta)$ be defined as in \eqref{eq:defm}. Then:
\renewcommand{\labelenumi}{\alph{enumi})}
\begin{enumerate}
\item $F_{\beta} \in C^3(-1,1)$ is an even function, i.e. $F_{\beta}(t)=F_{\beta}(-t)\,$ for all $t \in (-1, 1)$.
\item $F'_{\beta} < 0\,$ on $(0,m)\,$, $F'_{\beta} >0\,$ on $(m, 1)\,$, and $F''_{\beta}(m) > 0\,$.
\item For all $t \in (-1, 1)$ we have $e^{-F_{\beta}(t)/2} \leq e^{9\beta/2}(1-|t|)$.
\end{enumerate}
\end{proposition}

\begin{proof}
Statement a) is obvious. Statement b) follows from a straight forward computation for which it is useful to observe that
$\ln\frac{1+t}{1-t} =2$Artanh$(t)$ and $F'_{\beta}(t) = \frac{2}{\beta(1-t^2)}($Artanh$(t) -\beta t)$. Because of the evenness
of $F_{\beta}$ it suffices to prove statement c) for $0 \leq t < 1$ only. Set $X:= -\ln(1-t)\geq 0$. Since $\ln(1+t) \geq 0$ we have
\begin{eqnarray*}
\frac{F_{\beta}(t)}{2} = \frac{1}{8\beta}(X+\ln(1+t))^2 +\frac{1}{2}(\ln(1+t) - X) \geq \frac{1}{8\beta} X^2 -\frac{1}{2}X \,.
\end{eqnarray*}
Using in addition that $\frac{1}{8\beta} X^2 -\frac{3}{2}X + \frac{9}{2}\beta \geq 0$ the claim follows.
\end{proof}

We now have all ingredients to verify Lemma \ref{lemma:A5}.

\begin{proof}[Proof (Lemma \protect\ref{lemma:A5})]
We begin by evaluating the asymptotic behavior of the norming constant
\begin{eqnarray*}
Z_N = \int_{-1}^1 e^{-N^{\alpha} F_{\beta} (t)/2} \phi(t) \,dt \quad \mbox{with} \quad  \phi(t) =\frac{1}{1-t^2}.
\end{eqnarray*}
In order to apply Proposition \ref{prop:Laplace} we need to split the domain of integration into those four regions where
$F_{\beta}$ is monotone. Due to the evenness of $F_{\beta}$ (Proposition \ref{prop:F} a) it suffices to consider
the integrals over $[-m, 0]$ and $[m, 1)$. In both cases the parameters for condition 2. of Proposition \ref{prop:Laplace}
are $\kappa = 2$, $P= F''_{\beta}(-m)/2=F''_{\beta}(m)/2>0$, $\lambda=1$, and $Q=1/(1-m^2)$. Conditions 1. and 3. of
Proposition \ref{prop:Laplace} are satisfied because of statements b) and c) of Proposition \ref{prop:F}. Hence
\begin{equation}
\label{asymp:Z}
Z_N \approx \frac{2 \, \Gamma ({\textstyle\frac{1}{2})}}{1-m^2} \, \Big(\frac{4}{F^{''}_{\beta} (m) N^{\alpha}}\Big)^{\frac{1}{2}} e^{-N^{\alpha} F_{\beta} (m)/2} \, .
\end{equation}
Keeping the integrand fixed one may apply Proposition \ref{prop:Laplace} also to the integral over $[-\frac{m}{2}, 0]$ (now
$\kappa = 1$, $P= F'_{\beta}(-m/2)$, $\lambda=1$, $Q=1/(1-(m/2)^2)$). We arrive at
\begin{equation}
\label{asymp:a}
\int_{-\frac{m}{2}}^{\frac{m}{2}} e^{-N^{\alpha} F_{\beta} (t)/2} ~ \frac{dt}{1-t^2}
\approx \frac{2\, \Gamma (1)}{1-(\frac{m}{2})^2} \cdot \frac{2}{F'_{\beta} (-\frac{m}{2}) N^{\alpha}} \,
e^{-N^{\alpha} F_{\beta} (\frac{m}{2})/2}
\end{equation}
Taking the quotient of \eqref{asymp:a} and \eqref{asymp:Z} one easily derives statement a) with, say,
$\delta =\frac{1}{2}( F_{\beta} (\frac{m}{2}) - F_{\beta} (m)) > 0$.

In order to obtain the remaining claims
set $\phi_1(t) := |t-m|^{\ell}/(1-t^2)$ resp.
$\phi_2(t) := (1+m^2-2mt)^{\ell}/(1-t^2)$. Applying Proposition \ref{prop:Laplace} to the intervals $[-m, -\frac{m}{2}]$ and $[m,1)$
we find
\begin{eqnarray*}
 \int_{\frac{m}{2}}^1 e^{- N^{\alpha} F_{\beta} (t)/2} \phi_1(t) \, dt &\approx &
 \frac{\Gamma ({\textstyle\frac{\ell+1}{2}})}{1-m^2}  \,
 \Big(\frac{4}{F^{''}_{\beta} (m) N^{\alpha}}\Big)^{\frac{\ell+1}{2}} e^{-N^{\alpha} F_{\beta} (m)/2}\vspace{0.3cm}\\
 \int_{\frac{m}{2}}^1 e^{- N^{\alpha} F_{\beta} (t)/2} \phi_2(t) \, dt &\approx &
 \frac{\Gamma ({\textstyle\frac{1}{2}})}{(1-m^2)^{1-\ell}} \,
 \Big(\frac{4}{F^{''}_{\beta} (m) N^{\alpha}}\Big)^{\frac{1}{2}} e^{-N^{\alpha} F_{\beta} (m)/2}
\end{eqnarray*}
and statements b) and c) follow from \eqref{asymp:Z}.
\end{proof}

\section{Proof of the Main Result}\label{sec.3}

We begin the proof with large deviations estimates that demonstrate the efficiency of the indicators
$\mathcal{X}_{N, \pm}$ defined in \eqref{eq:A3.1}. They are immediate consequences of Hoeffding's inequality.
Nevertheless, we provide a short proof for the convenience of the reader.
\begin{proposition}
\label{prop:A4}

Let $P_t$ and $S_N$ be defined as in \eqref{def:Pt}, \eqref{eq:A3.1}.
For $a \in (0,1)$ denote $q_a := - \frac{1}{4} \log (1-a^2) > 0$.
Then the following estimates hold true.
\renewcommand{\labelenumi}{\alph{enumi})}
\begin{enumerate}
\item For all $t \in [a,1]: P^{(N^2)}_t (S_N \leq 0) \leq e^{-q_a N^2}$.
\item For all $t \in [-1,-a]: P^{(N^2)}_t (S_N > 0) \leq e^{-q_a N^2}$.
\end{enumerate}
\end{proposition}
\begin{proof}
Since $P_{\pm 1} (S_N = \pm \frac{1}{2} N (N+1)) = 1$ we only need to consider $t \in (-1,1)$.
Define $\lambda (t) := \frac{1}{2} \log \frac{1-t}{1+t}$. Then
\begin{equation*}
E_t (e^{\lambda (t) X_N(1,1)}) = \frac{1}{2} e^{\lambda (t)} (1+t) + \frac{1}{2} e^{- \lambda (t)} (1-t) = \sqrt{1 - t^2} \,.
\end{equation*}
For $t \in [a,1)$ we have $\lambda (t) < 0$. Using in addition that the random variables $X_N(i,j)$, $1 \leq i \leq j \leq N$,
are independent and identically distributed with respect to the probability measure $P_t$ we obtain
\begin{equation*}
P_t (S_N \leq 0) \leq E_t (e^{\lambda (t) S_N}) = E_t (e^{ \lambda (t) X_N (1,1)})^{\frac{N (N + 1)}{2}} \leq e^{-q_a N^2} \,.
\end{equation*}
Similarly, $\lambda (t) > 0$ for $t \in (-1,-a]$ and $P_t (S_N > 0) \leq E_t (e^{\lambda (t) S_N}) \leq e^{-q_a N^2}$.
\end{proof}

We are now ready to prove the lemma which is the key for proving our main result, i.e. to show that $Y_N$
defined in \eqref{eq:A4.1} has approximately uncorrelated entries.

\begin{lemma}
\label{lemma:A7}

Let $(X_N)_N$ be a generalized $\IP_{N^2}^{N^{\alpha} F_{\beta}}$- Curie-Weiss ensemble with
$\alpha \geq 1$ and $\beta > 1$ and let $Y_N$ be defined as in \eqref{eq:A4.1} $($see also \eqref{eq:A2.2}$)$.
Then $(Y_N)_N$ has approximately uncorrelated entries with respect to the probability measure
$\IP_{N^2}^{N^{\alpha} F_{\beta}}$ (see Definition \ref{def:approx}).

\end{lemma}
\begin{proof}

Let us first consider \eqref{eq:A7.2} and define (see \eqref{eq:A3.1}, \eqref{eq:A4.1})
\begin{align*}
G_{N,\pm} & := \mathcal{X}_{N,\pm} \prod^{\ell}_{\nu = 1} Y_N (i_{\nu}, j_{\nu})^2
= \mathcal{X}_{N,\pm} \prod^{\ell}_{\nu = 1} (\mathcal{X}_{N,\pm} Y_N (i_{\nu}, j_{\nu}))^2\\
&  = \mathcal{X}_{N,\pm} \prod^{\ell}_{\nu = 1} Y_{N,\pm} (i_{\nu}, j_{\nu})^2
 = (1-m^2)^{-\ell} \mathcal{X}_{N,\pm} \prod^{\ell}_{\nu = 1} (X_N (i_{\nu}, j_{\nu}) \mp m)^2
\end{align*}
Obviously,
\begin{equation*}
\IE_{N^2}^{N^{\alpha} F_{\beta}} \Big(\prod_{\nu = 1}^{\ell} Y_N (i_{\nu}, j_{\nu})^2\Big) =
\IE_{N^2}^{N^{\alpha} F_{\beta}} (G_{N,+}) + \IE_{N^2}^{N^{\alpha} F_{\beta}} (G_{N,-}) \quad \mbox{and}
\end{equation*}
\begin{equation}
\label{eq:A8.1}
\IE_{N^2}^{N^{\alpha} F_{\beta}} (G_{N,\pm}) = (1-m^2)^{-\ell} \int_{-1}^1 E_t (\mathcal{X}_{N,\pm} H_{N,\pm}) ~d \nu_N (t) \,,
\end{equation}
with $H_{N,\pm} := \prod^{\ell}_{\nu = 1} (X_N (i_{\nu}, j_{\nu}) \mp m)^2$.

Note that $X_N^2 (i_{\nu}, j_{\nu}) = 1$ and that $E_t (X_N(i_{\nu}, j_{\nu})) = t$.
Since the index pairs $(i_1,j_1), \ldots , (i_{\ell},j_{\ell})$ are assumed to be pairwise disjoint we obtain
\begin{equation}
\label{eq:A8.2}
E_t (H_{N,\pm}) = (1 + m^2 \mp 2mt)^{\ell} \,.
\end{equation}
In order to evaluate $\IE_{N^{\alpha}}^{N^2 F_{\beta}} (G_{N, +})$ we decompose
\begin{eqnarray}
\nonumber
\int_{-1}^1\!\! E_t (\mathcal{X}_{N,+} H_{N,+}) \,  d \nu_N \!\!&=&\! \!\!
\int_{-1}^{-\frac{m}{2}} E_t (\mathcal{X}_{N,+} H_{N,+}) \, d \nu_N  \\ \label{eq:A9.1}
&+&\!\!\! \int_{-\frac{m}{2}}^{\frac{m}{2}} E_t (\mathcal{X}_{N,+} H_{N,+}) \, d \nu_N \\ \nonumber
&-&\!\!\! \int_{\frac{m}{2}}^1\!  E_t (\mathcal{X}_{N,-} H_{N,+}) \, d \nu_N  + \int_{\frac{m}{2}}^1\! E_t (H_{N,+})  d \nu_N .
\end{eqnarray}
From Proposition \ref{prop:A4}, from Lemma \ref{lemma:A5} a), and from the trivial estimate $|H_{N,+}| \leq (1+m)^{2\ell}$
we obtain constants $q, C, \delta > 0$ such that for all $N \in \IN$:
\begin{eqnarray*}
\Big{|}\int_{-1}^{- \frac{m}{2}} E_t (\mathcal{X}_{N,+} H_{N,+}) \,  d \nu_N \Big{|}
& \leq & (1+m)^{2\ell} e^{-q N^2} \,, \\
\Big{|}\int_{\frac{m}{2}}^{1} E_t (\mathcal{X}_{N,-} H_{N,+}) \,  d \nu_N \Big{|}
& \leq & (1+m)^{2\ell} e^{-q N^2} \,, \\
\Big{|}\int_{-\frac{m}{2}}^{\frac{m}{2}} E_t (\mathcal{X}_{N,+} H_{N,+}) \,  d \nu_N \Big{|}
& \leq & C (1+m)^{2\ell} e^{- \delta N^{\alpha}} .
\end{eqnarray*}
Thus the first three summands on the right hand side of \eqref{eq:A9.1} all converge to zero as $N$ tends to $\infty$.
By \eqref{eq:A8.1}, \eqref{eq:A8.2} and Lemma \ref{lemma:A5} c) we have proved
\begin{equation*}
\underset{N \rightarrow \infty}{\lim} \IE_{N^2}^{N^{\alpha} F_{\beta}} (G_{N,+}) = \frac{1}{2} \,.
\end{equation*}
By analogue arguments the expected value of $G_{N,-}$ can be seen to converge to $\frac{1}{2}$ as well
and \eqref{eq:A7.2} is established.
Estimate \eqref{eq:A7.1} can be proved in a similar fashion.
To this end redefine
\begin{equation*}
H_{N,\pm} := \prod^{\ell}_{\nu = 1} (X_N(i_{\nu}, j_{\nu}) \mp m) \prod^p_{\rho = 1} (X_N(u_{\rho}, v_{\rho}) \mp m) \,.
\end{equation*}
The expected value in \eqref{eq:A7.1} equals
\begin{equation}
(1-m^2)^{-\frac{\ell+p}{2}} \Big(\int_{-1}^1 E_t (\mathcal{X}_{N,+} H_{N,+}) \, d \nu_N  +
\int_{-1}^1 E_t (\mathcal{X}_{N,-} H_{N,-}) \,  d \nu_N \Big) .
\end{equation}
Oberve that
\begin{equation*}
E_t (H_{N,\pm}) = (t \mp m)^{\ell} \cdot Q_{\pm} (t)
\end{equation*}
with $Q_{\pm}$ being a polynomial of degree $\leq p$ and $|Q_{\pm} (t)| \leq (1+m)^p$ for all $t \in [-1,1]$.
Using \eqref{eq:A9.1}, the arguments thereafter, $|H_{N,+}| \leq (1+m)^{\ell+p}$, and statement b) of Lemma \ref{lemma:A5}
one may establish the existence of constants $C, \delta, q > 0$ such that for all $N \in \IN:$
\begin{equation*}
\Big{|}\int_{-1}^{1} E_t (\mathcal{X}_{N,+} H_{N,+}) \,  d \nu_N \Big{|} \leq
(1+m)^{\ell+p} (2 e^{-q N^2} + C e^{- \delta N^{\alpha}} + C_{\ell} N^{- \frac{\alpha \ell}{2}}) .
\end{equation*}
Clearly, the same estimate holds for $|\int_{-1}^1 E_t (\mathcal{X}_{N,-} H_{N,-}) \, d \nu_N|$
and \eqref{eq:A7.1} is proved since we have assumed $\alpha \geq 1$.

\end{proof}

The final observation we need is that rank $1$ perturbations preserve the weak convergence (in probability) of the
eigenvalue distribution measures (cf.\;\cite{FuerediK}).

\begin{lemma}
\label{lemma:perturb}
Let $(\Omega, \mathcal{F}, \IP)$ denote the probability space underlying two ensembles $(A_N)_N$, $(B_N)_N$
of real symmetric $N \times N$ matrices that satisfy rank$(A_N(\omega)-B_N(\omega)) \leq 1$
for all $\omega \in \Omega$.
Denote by $\sigma_N$ resp. by $\mu_N$ the eigenvalue distribution measures of $A_N$ resp. $B_N$. Assume furthermore that
$(\mu_N)_N$ converges weakly in probability to the semicircle law $\sigma_{sc}$.
Then  $(\sigma_N)_N$ also converges weakly in probability to the semicircle law.
\end{lemma}

\begin{proof}
The assumption on the rank of $A_N-B_N$ implies that the eigenvalues of the two matrices interlace.
Therefore we have for any interval $I \subset \mathbb{R}$ and any $\omega \in \Omega$ that
\begin{equation}
\label{eq:A11.1}
|\#(\text{eigenvalues of } A_N(\omega) \text{ in } I) - \#(\text{eigenvalues of } B_N(\omega) \text{  in } I)| \leq 2 \,.
\end{equation}
In order to prove Lemma \ref{lemma:perturb} we fix a bounded continuous function $f: \mathbb{R} \rightarrow \mathbb{R}$
and numbers $\varepsilon, \gamma > 0$. We have to show that there exists $N_0$ such that for all $N \geq N_0$ the estimate
\begin{equation}
\label{eq:A12.1}
\IP \bigg{(} \Big{|} \int f ~ d \sigma_N - \int f ~  d \sigma_{sc} \Big{|} > \varepsilon \bigg{)} < \gamma
\end{equation}
holds. Choose a step function $g = \sum_{i = 1}^r \alpha_i \mathcal{X}_{I_i}$ $ (\alpha_i \in \IR$,
the intervals $I_i \subset [-4,4]$ are pairwise disjoint) such that
\begin{equation}
\label{eq:A12.2}
\sup \{|f (x) - g (x)| : x \in [-4,4] \} \leq \min \big(\frac{\varepsilon}{6},1\big) \,.
\end{equation}			
Write $\int f ~  d \sigma_N - \int f ~ d \sigma_{sc} = \sum_{j=1}^6 \triangle_j$, with
\begin{align*}
\triangle_1 & := \int f ~  d \mu_N - \int f ~ d \sigma_{sc}\,, ~~~~~~~~~  \triangle_2 := \int_{-4}^4 (f-g) ~ d \sigma_N\,, \\
\triangle_3  &:= \int_{-4}^4 g ~  d \sigma_N - \int_{-4}^4 g ~ d \mu_{N}\,, ~~~~~
\triangle_4 := \int_{-4}^4 (g-f) ~  d \mu_N \,,  \\
\triangle_5 &:= \int_{\IR \setminus [-4,4]} f ~ d \sigma_N\,, ~~~~~~~~~~~~~~~~~~~
\triangle_6 := -\int_{\IR \setminus [-4,4]} f ~  d \mu_N \,.
\end{align*}
It suffices to show for each $j = 1, \ldots , 6$ that there exists $N_j$ such that for all $N \geq N_j$ we have
\begin{equation}
\label{eq:A12.3}
\IP (|\triangle_j| > \frac{\varepsilon}{6}) < \frac{\gamma}{6}
\end{equation}
Indeed, \eqref{eq:A12.1} then holds for all $N \geq N_0 := \max \{N_j : j = 1, \ldots , 6 \}$.

For $j = 2,4$ estimate \eqref{eq:A12.3} follows trivially from \eqref{eq:A12.2} with $N_j = 1$.
The case $j = 1$ is a consequence of the assumption of Lemma \ref{lemma:perturb}.
Estimate \eqref{eq:A11.1} together with \eqref{eq:A12.2} imply
\begin{equation*}
|\triangle_3| \leq \sum_{i=1}^r |\alpha_i| \cdot |\sigma_N (I_i) - \mu_N (I_i)|
\leq r (||f||_{\infty} + 1) \cdot \frac{2}{N} < \frac{\varepsilon}{6}
\end{equation*}
by choosing $N_3$ sufficiently large. Next we treat $j = 6$.
Using the assumed weak convergence of $(\mu_N)_N$ again and choosing an arbitrary continuous test function
$f_0 : \IR \rightarrow [0,1]$ with ${f_0}\vert_{[-2,2]} = 0$ and ${f_0}\vert_{\IR \setminus [-4,4]} = 1$
one may obtain $N_6$ such that for all $N \geq N_6:$
\begin{equation}
\label{eq:A13.1}
\IP \Big( \mu_N (\IR \setminus [-4,4]) > {\textstyle \frac{\varepsilon}{12 (||f||_{\infty} + 1)}}\Big) \leq
\IP \Big(\int f_0 \, d\mu_N  > {\textstyle \frac{\varepsilon}{12 (||f||_{\infty} + 1)}}\Big) < \frac{\gamma}{6} .
\end{equation}
On the one hand this implies $\IP (|\triangle_6| > \frac{\varepsilon}{12}) < \frac{\gamma}{6}$ for all  $N \geq N_6$.
On the other hand we may use \eqref{eq:A13.1} to handle $\triangle_5$. It follows from \eqref{eq:A11.1} that
\begin{equation*}
|\mu_N (\IR \setminus [-4,4]) - \sigma_N ( \IR \setminus [-4,4])|
\leq \frac{2}{N} < \frac{\varepsilon}{12 ( ||f||_{\infty} + 1)}
\end{equation*}
for all $N \geq \tilde{N}_5$ with a suitable choice of $\tilde{N}_5$.
Setting $N_5 := \max (\tilde{N}_5, N_6)$ then completes the proof with the help of \eqref{eq:A13.1}.
\end{proof}

After all ingredients have been gathered we may now conclude the proof of our main result.

\begin{proof}[Proof (Theorem \protect\ref{theo:A1})]
We denote by $\sigma_N$ resp. $\mu_N$ the eigenvalue distribution measures of $A_N :=  X_N/\sqrt{N(1-m^2)}$,
resp. $B_N := Y_N/\sqrt{N}$ (see \eqref{eq:A4.1}, \eqref{eq:A2.2} for the definition of $Y_N$).
Since $(Y_N)_N$ is an approximately uncorrelated scheme of random variables it follows from \cite[Theorem 5]{HKW}
that $(\mu_N)_N$ converges weakly in probability to the semicircle law $\sigma_{sc}$.
Since
\begin{equation*}
A_N - B_N = \frac{m}{\sqrt{N(1-m^2)}} (\mathcal{X}_+ - \mathcal{X}_-) \mathcal{E}_N
\end{equation*}
always has rank 1, the claim now follows from Lemma \ref{lemma:perturb}.
\end{proof}

\begin{remark}\label{remark:generalF}
Observe that in the proof of Theorem \ref{theo:A1} we did not use the special form of the function $F_{\beta}$
but only those properties that were collected in Proposition \ref{prop:F}. Hence the result also holds for
generalized $\IP_{N^2}^{N^{\alpha} F}$- Curie-Weiss ensembles (cf. \cite[Definition 29]{HKW}), if $F : (-1,1) \to
\mathbb{R}$ satisfies all the properties stated in Proposition \ref{prop:F}. Of course, condition c) may by replaced by the
requirement that $\int_{-1}^1e^{-x F (t)/2}\frac{dt}{1-t^2} < \infty$ for sufficiently large values of $x$,
because this is the only use of property c).
\end{remark}

\texttt{\bigskip }

\bigskip\bigskip

\noindent
\begin{tabular}{lcl}
\textbf{Werner Kirsch}&\quad& \texttt{werner.kirsch@fernuni-hagen.de}\\
\textbf{Thomas Kriecherbauer}&\quad& \texttt{thomas.kriecherbauer@uni-bayreuth.de}
\end{tabular}
\end{document}